\documentclass[pra,twocolumn,superscriptaddress,
preprintnumbers,amsmath,amssymb,floatfix]{revtex4-2}
\usepackage{enumerate}
\usepackage{kotex}
\usepackage{graphicx}
\usepackage{float}
\usepackage{subfigure}
\usepackage{amsthm}
\usepackage{tensor}
\usepackage{color}
\usepackage[all]{xy}
\usepackage{tikz}
\usepackage{dsfont}
\usepackage{times,txfonts}
\usetikzlibrary{positioning}
\usepackage{braket}
\usepackage{bbm}

\newtheorem{Thm}{Theorem}

\theoremstyle{definition}

\begin{document}


\title{Quantum teleportation with partially entangled joint measurements induced by coherent errors}

\author{Jeonghyeon Shin} 
\affiliation{Center for Quantum Technology, Korea Institute of Science and Technology (KIST), Seoul 02792, Korea}
\affiliation{Department of Mathematics and Research Institute for Basic Sciences, Kyung Hee University, Seoul 02447, Korea}

\author{Jaehak Lee} 
\affiliation{Center for Quantum Technology, Korea Institute of Science and Technology (KIST), Seoul 02792, Korea}
\affiliation{Division of Quantum Information, KIST School, Korea University of Science and Technology (UST), Seoul 02792, Korea}

\author{Soojoon Lee} 
\affiliation{Department of Mathematics and Research Institute for Basic Sciences, Kyung Hee University, Seoul 02447, Korea}
\affiliation{School of Computational Sciences, Korea Institute for Advanced Study, Seoul 02455, Korea}

\author{Seung-Woo Lee}
\email{swleego@gmail.com}
\affiliation{Department of Physics, Pohang University of Science and Technology (POSTECH), Pohang 37673, Korea}

\date{\today}

\begin{abstract}
Quantum teleportation is a fundamental protocol in quantum information science, whose performance is conventionally evaluated under the assumption of ideal Bell-state measurements. In realistic implementations, however, joint measurements are often imperfect and can deviate from maximally entangled bases due to coherent errors in entangling operations. In this work, we analytically show how the entanglement of joint measurements determines teleportation performance and propose a strategy to overcome the limitations imposed by partially entangled joint measurements to recover the unit teleportation fidelity with a nonzero success probability. 
This is achieved by optimizing a reversing operation on the teleported state to compensate for the known coherent deformation of the Bell-state measurement.
We then derive an exact equation revealing a quantitative relation between measurement entanglement, channel entanglement, and the success probability to realize the unit-fidelity two-qubit teleportation.
We illustrate our results using elegant joint measurements and realistic coherent error models arising from imperfect entangling operations in quantum systems. Our work provides fundamental insight into the role of measurement entanglement in quantum teleportation and 
offers a practical framework for achieving faithful teleportation without requiring substantial modifications to existing hardware.
\end{abstract}
\maketitle

\section{Introduction}

Quantum teleportation is a key protocol in quantum information science, enabling the faithful transfer of an unknown quantum state using shared entanglement and classical communication \cite{Bennett1993, NielsenChuang, Pirandola2015}. Since its original proposal \cite{Bennett1993}, teleportation has been extensively studied and experimentally demonstrated not only as a fundamental primitive \cite{Bouwmeester97, Boschi98, furu, Barrett04, Olmschenk09, Steffen13, wang15} but also as a building block for quantum computing and quantum communication architectures \cite{Gottesman1999, Briegel1998, Raussendorf2001, Kimble2008, Riedmatten04, Lee13, Sum16, Valivarthi16, Ren17, Lee19, Valivarthi20, Llewellyn20, Chou18, Wan19, Lee20}. Numerous works have been focused on optimizing teleportation performance in the presence of noise and imperfections arising in realistic implementations \cite{Nielsen96, Horodecki1998, Banaszek00, Karlsson98, JLee00, Mor99, WSon00, Bowen2001, Li00, Oh2002, Agrawal2002, Verstraete2003, kimin12, Lee2021, Im20, Brauer2024}. 
In most existing studies, the degradation of performance is primarily attributed to imperfections and noise in the shared entangled channel, while the entangled joint measurement performed by the sender is frequently assumed to be an ideal projection onto the maximally entangled basis, i.e.,~an ideal Bell-state measurement (BSM).
However, practical implementations of teleportation inevitably deviate from the ideal BSM. In realistic quantum systems, joint measurements are realized using finite-depth circuits composed of imperfect entangling gates, and are thus subject to coherent errors \cite{Knill2005, Greenbaum2017, Bravyi2018, Venn2023}, such as unwanted ZZ interactions in superconducting qubits \cite{Ku2020, Ni2022, Eckstein2024, Sundaresan2020} or residual XX-type couplings in neutral-atom and ion-trap systems \cite{Morgado2021, Fang2022}. Unlike stochastic noise including decoherence, these coherent errors deform the measurement basis itself, effectively transforming an ideal BSM into a partially entangled one. As a result, such non-ideal implementations of BSMs can crucially degrade teleportation performance. This reveals that teleportation performance is determined not only by the entanglement of the quantum channel but also by the entanglement of the joint measurement itself.

Recent works have highlighted that entangled joint measurements are not merely passive ingredients but active quantum resources with nontrivial structure \cite{Gisin2019, Cavalcanti2023, Baumer2021, DelSanto2024, Pauwels2025}. Various families of non-maximally entangled joint measurements, including elegant joint measurements (EJM) \cite{Tavakoli2021}, have been investigated in diverse contexts such as quantum nonlocality \cite{Tavakoli2021, Tavakoli2022}, entanglement swapping \cite{Huang2022}, and quantum teleportation \cite{Ding2024}. Despite this growing interest, the role of measurement entanglement in quantum teleportation remains largely unexplored. In particular, it is still unclear how a reduction in the entanglement of the joint measurement quantitatively affects teleportation performance, and whether such degradation is fundamentally irreversible.
From this perspective, investigating teleportation performance under imperfectly entangled joint measurements is not only of fundamental theoretical interest but also of direct relevance to practical implementations required in most current quantum technologies.

In this work, we investigate how the entanglement of partially entangled joint measurements induced by coherent errors determines teleportation performance and propose a strategy to overcome the limitations imposed by partially entangled joint measurements and recover the unit teleportation fidelity with finite probability. We first show that, within the standard quantum teleportation protocol, any reduction in measurement entanglement leads to an unavoidable decrease in teleportation fidelity, even when the shared quantum channel is perfectly entangled. This result highlights not only the impact of non-maximally entangled joint measurements but also an intrinsic limitation of the conventional deterministic teleportation framework. We then demonstrate that this limitation can be overcome by adopting a measurement-reversal (MR) framework \cite{Lee2021}, in which the receiver performs an optimized probabilistic reversing operation conditioned on the measurement outcome. Remarkably, for pure-state teleportation, we show that unit-fidelity teleportation can always be recovered regardless of how weakly entangled the joint measurement is, provided that the measurement imperfection is known.

Moreover, in the two-qubit case, we derive an equation that quantitatively characterizes how the success probability of achieving unit-fidelity teleportation depends on both the entanglement of the joint measurement and the entanglement of the quantum channel. For higher-dimensional systems, we derive bounds for the maximally entangled channel case that quantify how entanglement of the joint measurement constrains the success probability. This constitutes a central result of our work and reveals a clear trade-off between measurement entanglement, channel entanglement, and the probability of faithful teleportation.
We then apply our analytical framework to two classes of examples. (i) First, we analyze EJMs, explicitly comparing the degradation of standard teleportation fidelity with the heralded performance achievable within the MR framework. (ii) We also consider two representative realistic coherent error models such as ZZ interactions during time-continuous realizations of BSMs in superconducting qubits and residual XX-type couplings caused in neutral-atom and trapped-ion systems. These illustrate how imperfect BSM implementations manifest as a reduction of measurement entanglement and how their effects can be mitigated through optimized reversal strategies.
We believe that our work provides fundamental insight into the role of measurement entanglement in quantum teleportation, as well as a practical tool for achieving faithful teleportation without requiring substantial or complex modifications to the underlying hardware. It thus offers a practical framework for addressing coherent errors in large-scale quantum processors and quantum networks, where quantum teleportation serves as a fundamental building block.

\section{Effect of coherent error in quantum teleportation} 
\label{sec:effect of coherent error}

\begin{figure}[t]
\centering
\includegraphics[width=8.5cm]{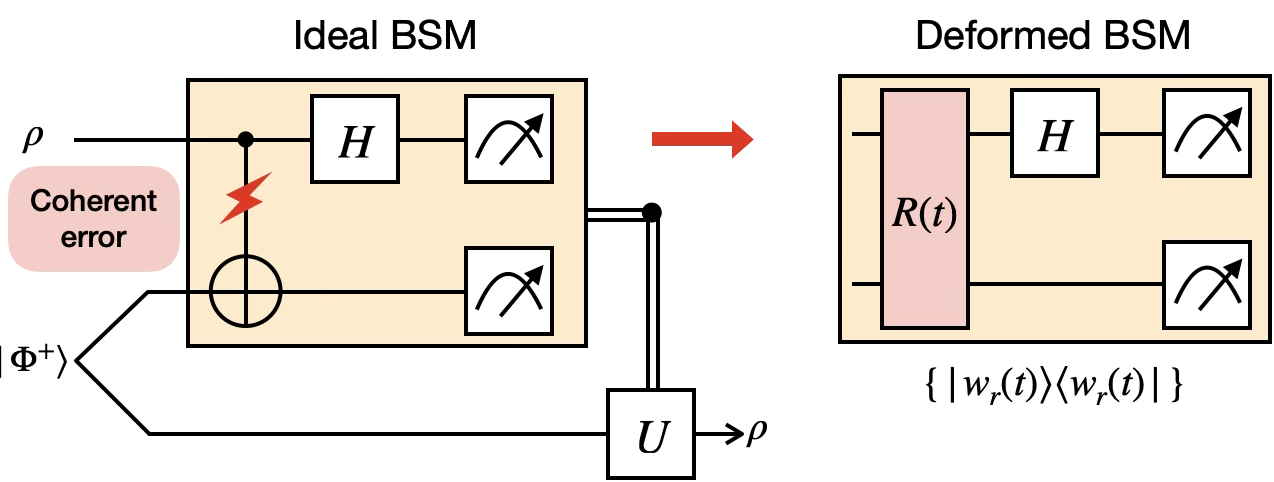}
\caption{Coherent error in a Bell-state measurement (BSM). In the ideal circuit, the BSM is implemented using a CNOT gate followed by a Hadamard gate and projective measurements. A coherent error in the entangling gate deforms the ideal Bell basis into a rotated, partially entangled measurement basis parameterized by the error strength $t$.}
\label{fig:deformed BSM}
\end{figure}

The quantum teleportation protocol originally proposed in Ref.~\cite{Bennett1993} (hereafter referred to as the \textit{standard} protocol) enables the deterministic transfer of an unknown quantum state with unit fidelity under ideal conditions. This protocol relies on three key ingredients. First, Alice (sender) and Bob (receiver) share an entangled quantum channel, which is ideally a maximally entangled state. Second, Alice performs a maximally entangled joint measurement, i.e.,~an ideal Bell-state measurement (BSM) on the input state and her share of the quantum channel. Third, Alice communicates the measurement outcome to Bob via a classical channel, allowing Bob to apply an appropriate unitary correction and recover the original input state. A typical circuit model of the standard teleportation protocol is illustrated in Fig.~\ref{fig:deformed BSM}.

In realistic implementations, however, the joint measurement as well as the shared entangled channel can deviate from maximal entanglement due to imperfections. Conventionally, teleportation performance has been evaluated under the assumption that Alice performs an ideal BSM, while the effects of imperfect and noisy quantum channels have been extensively studied. In practice, a BSM is not a primitive operation but is realized through a finite-depth circuit composed of entangling gates and local rotations.
Consequently, imperfections in the underlying entangling operations can induce coherent errors that deform the measurement basis away from the ideal Bell basis as shown in Fig.~\ref{fig:deformed BSM}. 
Such a deformed measurement basis can be written as $\ket{w_i(\theta)}=U(\theta)\ket{\beta_i}$, where $\ket{\beta_i}$ denotes the ideal Bell basis for $i=0,1,2,3$. Unlike stochastic noise such as decoherence, which primarily randomizes outcome statistics, coherent errors modify the entanglement structure of the joint measurement itself and thereby directly degrade teleportation performance.

\begin{figure}[t]
\centering
\includegraphics[width=8.5cm]{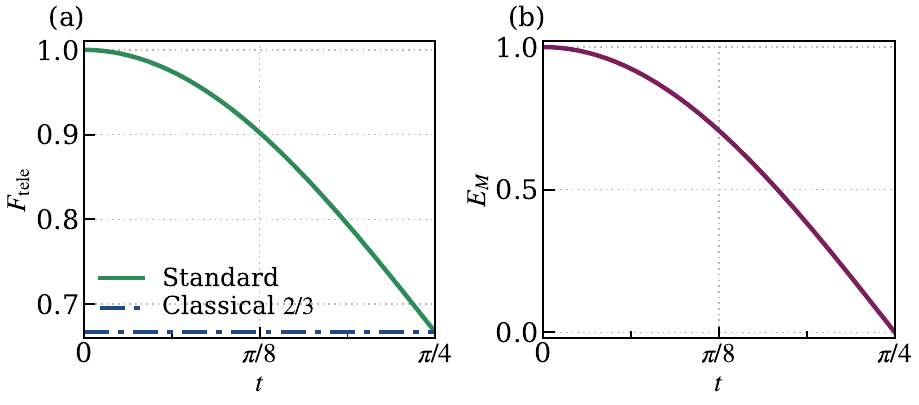}
\caption{Effect of coherent error in BSM on standard teleportation performance. (a) Average teleportation fidelity $F_\text{tele}$ decreases as imperfect entangling time $t$ increases. The horizontal line marks the classical limit $2/3$. (b) The measurement entanglement $E_M$, quantified by concurrence, also decreases monotonically, showing the fidelity loss directly follows from the reduction of measurement entanglement.}
\label{fig:coherent teleportation}
\end{figure}

To illustrate the effect of coherent errors in the simplest example, let us consider a rotated time-dependent Bell basis,
\begin{equation}
\label{eq:deformedBB}
    \begin{aligned}
        \ket{w_0(\theta)} & = \sin \theta\ket{00}+i\cos \theta\ket{11} \\
        \ket{w_1(\theta)} & = \sin \theta\ket{01}+i\cos \theta\ket{10}\\
        \ket{w_2(\theta)} & = \cos \theta\ket{00}-i\sin \theta\ket{11}\\
        \ket{w_3(\theta)} & = \cos \theta\ket{01}-i\sin \theta\ket{10},
    \end{aligned}
\end{equation}
with $\theta = \pi/4-t$. This captures a class of imperfect BSMs arising from coherent rotation of the measurement basis, which naturally occur in realistic systems, e.g.,~due to residual $XX$-type Ising interactions during entangling operations (a more detailed analysis will be presented in Sec.~\ref{sec:applications_realistic}). In this model, an ideal entangling operation is implemented at $\theta =\pi/4$, while coherent errors lead to an imperfect entangling evolution of the form $R_{XX}(t) = e^{-i(\pi/4-t)X_1X_2}$. Here, $0\le t\le \pi/4$ quantifies the imperfect entangling time. Thus, increasing $t$ from $t=0$, which corresponds to the ideal BSM, continuously reduces the entanglement of the deformed Bell basis. This reduction can be quantified by the concurrence $E_M(t) = \cos 2t$, which decreases monotonically with $t$ as shown in Fig.~\ref{fig:coherent teleportation}(b).

Assuming teleportation is performed using the deformed Bell basis in Eq.~\eqref{eq:deformedBB} through a maximally entangled quantum channel, the average teleportation fidelity in the standard protocol is given by
\begin{equation}
    \text{(standard)}:~~F_\text{tele}(t) = \frac{2+\cos 2t}{3}= \frac{2+E_M(t)}{3},
\end{equation}
which decreases monotonically with $t$ as shown in Fig. \ref{fig:coherent teleportation}(a). The fidelity is unity at $t=0$ and decreases to the classical limit $2/3$ at $t=\pi/4$. This result clearly demonstrates that coherent errors in BSM can significantly degrade teleportation performance, even when the quantum channel remains perfectly maximally entangled within the standard teleportation protocol. 

\section{Teleportation protocol in measurement-reversal framework} 
\label{sec:framework}

Let us now introduce a generalized teleportation protocol within the measurement-reversal (MR) framework~\cite{Lee2021, Im20}, which will serve as the primary framework throughout this work. Within this framework, we allow for arbitrary joint measurements performed at Alice’s side, rather than restricting to an ideal BSM. Unlike the standard protocol, where Bob’s correction is restricted to unitary operations, the MR framework allows a more general probabilistic reversing operation at Bob’s side, enabling recovery of the original input state. A schematic of the teleportation protocol in the MR framework is illustrated in Fig.~\ref{fig:framework}. We consider a scenario in which Alice aims to teleport an unknown input state 
$\rho_{\bar{a}}$
through a shared quantum channel $\Phi_{ab} = \ket{\Phi}\bra{\Phi}$,
where $\ket{\Phi} = \sum_{i,j=0}^{d-1}E_{ij}\ket{i}\ket{j}$. Here, $E = [E_{ij}] \in \mathbb{C}^{d \times d}$ denotes the coefficient matrix of the bipartite state $\ket{\Phi}$, satisfying the normalization condition $\mathrm{Tr}(E^\dagger E) = 1$. Assume that Alice performs an arbitrary rank-one joint measurement $\{\ket{w_r}\bra{w_r}\}_{r=0}^{d^2-1}$ on the system $\bar{a}a$. We denote by $W_r\in\mathbb{C}^{d\times d}$ the coefficient matrix of $\ket{w_r}$, satisfying the condition $\mathrm{Tr}(W_r^\dagger W_s)=\delta_{rs}$.
Here, we allow both the shared quantum channel and the joint measurement to be partially entangled. 

Conditioned on the measurement outcome $r$, the (unnormalized) post-measurement state at Bob’s side is given by 
$\tilde{\rho}_b^r $
where
\begin{equation}
\label{eq:overall karus operator}
    \tilde{\rho}_b^r = (\bra{w_r}_{\bar{a}a}\otimes I_b)(\rho_{\bar{a}}\otimes \Phi_{ab})(\ket{w_r}_{\bar{a}a}\otimes I_b) =M_r\rho_{\bar{a}}M_r^\dagger,
\end{equation}
where $M_r\equiv E^TW_r^\dagger$
captures the combined effect of the shared quantum channel and the joint measurement on Bob’s system (see Appendix~\ref{A:post-measurement} for the detailed calculation).

This naturally identifies an outcome-dependent Kraus operator acting on the input state, given by $M_r \equiv E^TW_r^\dagger$. Accordingly, the overall process—including the local joint measurement, the shared nonlocal channel, and classical communication—can be described by a quantum instrument $\mathbf{M}_{\bar{a}\to b}=\{M_r\}$ satisfying the completeness relation $\sum_r M_r^\dagger M_r = \mathds{1}$, as illustrated in Fig.~\ref{fig:framework}. 
In the ideal setting where both the channel and the joint measurement are maximally entangled, $E$ and $W_r$ are proportional to unitary operators, and the map reduces (up to known local unitaries) to the standard teleportation. In contrast, deviations from maximal entanglement manifest as nonunitary distortions in $M_r$.
To complete the teleportation protocol, Bob applies an outcome-dependent reversing operation $R_r$ to recover the input state $\rho_{\bar{a}}$. Since $M_r$ is generally nonunitary, we allow a probabilistic recovery for each outcome $r$, implemented by a local filter operation $R_r$ heralded by a successful outcome. 

\begin{figure}[t]
\centering
\includegraphics[width=8.0cm]{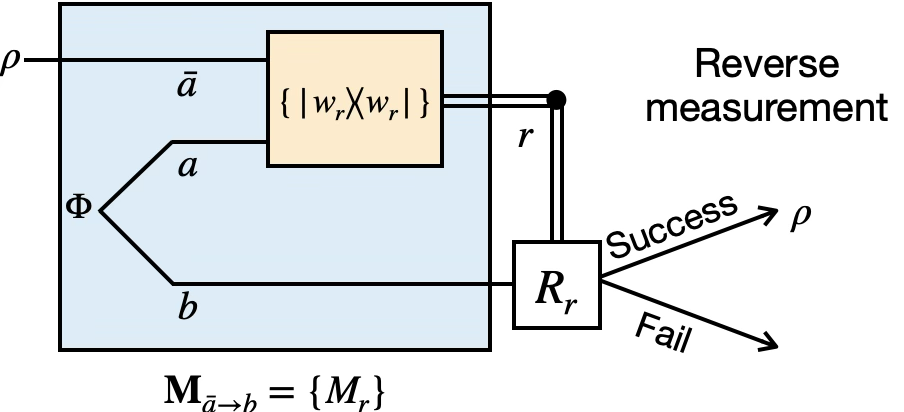}
\caption{Teleportation protocol in the measurement-reversal (MR) framework. Alice holds modes $\bar{a}$ and $a$, while Bob holds mode $b$. Alice performs a rank-one joint measurement on $\bar{a}a$ and sends the outcome $r$ to Bob via a classical channel. Bob then applies a reversing operation $R_r$ conditioned on $r$ to recover the input state $\rho_{\bar{a}}$ on mode $b$.}
\label{fig:framework}
\end{figure}

We quantify the performance of the protocol in terms of the heralded teleportation fidelity and the corresponding success probability.
Consider an arbitrary pure input state $\rho = \ket{\phi}\bra{\phi}$, with $\ket{\phi}$ sampled uniformly with respect to the Haar measure.
The heralded teleportation fidelity is defined as the average overlap between the input and the output state conditioned on the successful recovery, i.e.,
\begin{equation}
F_{\text{tele}}= 
    \frac{1}{P_\text{succ}}\int d\phi\sum_rp_{\text{succ},\phi}^r\braket{\phi|\,\rho^{r}_{\text{succ},\phi}|\phi}
\end{equation}
where $p_{\mathrm{succ},\phi}^{r}$ denotes the joint probability that Alice obtains outcome $r$ and Bob's reversing operation succeeds for the input state $\ket{\phi}$, and $\rho_{\mathrm{succ},\phi}^{r}$ is the corresponding normalized output state conditioned on successful reversal. The corresponding average success probability is given by
\begin{equation}
    P_\text{succ} = \int d\phi \sum_r p_{\text{succ},\phi}^r.
\end{equation}
The performance is jointly characterized by the heralded teleportation fidelity $F_{\mathrm{tele}}$ and the success probability $P_{\mathrm{succ}}$. Notably, successful reversal events can be identified and postselected, allowing the conditional fidelity to reach unity at the cost of a reduced success probability, while failure events are discarded rather than treated as delivered teleportation outputs.

Our protocol aims to achieve faithful teleportation, i.e.~unit teleportation fidelity, while maximizing the success probability $P_{\text{succ}}$. The performance of this protocol is fundamentally subject to the no-cloning constraint.
By the no-cloning theorem, information about the input state can either be successfully transferred to Bob--characterized by the teleportation fidelity and success probability--or leaked to Alice~\cite{Cheong2012, Bana, Lim14, SLee21, Hong22}. For this Haar-uniform ensemble of pure input states, we quantify the information about the input state available to Alice using the estimation fidelity, denoted by $L_\text{Alice}$, defined as 
\begin{equation}
    L_\text{Alice} = \int d\phi \sum_{r}p_\phi^r|\braket{\hat{\phi}_r|\phi}|^2,
\end{equation}
where $\ket{\hat{\phi}_r}$ is the estimated state conditioned on outcome $r$ based on her knowledge of the quantum instrument $\mathbf{M}_{\bar{a}\to b}$. 
This constraint manifests as a fundamental trade-off between the success probability $P_\text{succ}$ and the leaked information $L_\text{Alice}$ induced by $\mathbf{M}_{\bar{a}\to b}$, given by 
\begin{equation}
d(d+1)L_{\text{Alice}}^{\max}+(d-1)P_{\text{succ}}^{\max} \le 2d,
\end{equation}
where the maximum is evaluated over all possible strategies. This relation sets a fundamental bound on teleportation performance, implying that increased information leakage to Alice necessarily reduces the achievable success probability of faithful teleportation~\cite{Cheong2012, Bana, Lim14, SLee21, Hong22}.

Following Ref. \cite{Lee2021}, we show that teleportation described by the quantum instrument $\mathbf{M}_{\bar{a}\to b}$, arising from a pure entangled channel and a joint measurement, can be optimized to achieve unit fidelity in the successful events, i.e., $F_\text{tele}=1$, with a maximum success probability $P_\text{succ}^{\max}$ that saturates the fundamental no-cloning bound.
To construct the optimal reversing operator, consider the quantum instrument $\mathbf{M}_{\bar{a}\to b}=\{M_r\}$, and write each Kraus operator in its singular value decomposition as $M_r = P_r\Sigma_r Q_r^\dagger$, where $P_r$ and $Q_r$ are unitary operators and $\Sigma_r=\text{diag}(\sigma_0^r,\ldots,\sigma_{d-1}^r)$ with $\sigma_0^r\ge \cdots \ge\sigma_{d-1}^r$.
For each outcome $r$ with $\sigma_{\min}^r>0$, the optimal reversing operator is given by $R_r= \sigma_{\min}^r Q_r \Sigma_r^{-1}P_r^\dagger$ such that 
\begin{equation}
\label{eq:optimal reversing}
    R_rM_r = \sigma_{\min}^r \mathds{1}_{\bar{a}\to b},
\end{equation}
showing that the input state is perfectly recovered upon successful reversal~\cite{Cheong2012, SLee21}.
The joint probability of obtaining outcome $r$ and successfully recovering the input state is therefore given by $p_{\text{succ},\phi}^r = (\sigma_{\min}^r)^2$, which is independent of the input state $\rho$.
Consequently, an arbitrary input state $\rho$ can be faithfully teleported with the maximum success probability $P_\text{succ}^{\max} = \sum_r (\sigma_{\min}^r)^2$, while the corresponding maximum information leakage over all possible strategies is obtained as $L_\text{Alice}^{\max} = (d+\sum_r (\sigma_{\max}^r)^2)/d(d+1)$. 

Note that this construction requires Bob to know the effective Kraus operators \(M_r\), or equivalently, the coherent deformation of Alice’s implemented joint measurement. Such information can be obtained through prior calibration, measurement characterization, or Hamiltonian-parameter estimation of the BSM circuit. Imperfect estimates still allow the construction of a reversing operation, but no longer guarantee exact unit fidelity as we analyze in Sec.~\ref{sec:applications}.

\section{Role of joint measurement and quantum channel in faithful teleportation}

We analyze how the shared quantum channel and the joint measurement jointly determine the teleportation performance within our framework, using the maximum success probability of faithful teleportation as the operational benchmark for a given pair of resources.

We first consider the qubit ($d=2$) teleportation, where a complete analytical characterization is possible.
Before stating the theorem, we briefly outline its derivation. For each outcome $r$, the effective Kraus operator factorizes as $M_r = E^{T}W_r^\dagger$, and the maximum success probability for that outcome is determined by the smallest singular value of $M_r$. Thus, the problem reduces to evaluating $\sigma_{\min}(M_r)$ in terms of the channel coefficient matrix $E$ and the measurement coefficient matrix $W_r$ (See Appendix \ref{sec: proof of thm1} for details).

\begin{Thm}
\label{thm:success probability}
Assume $d=2$. Let $E_c$ be the concurrence of the shared two-qubit channel $\ket{\Phi}$, and let $E_r$ be the concurrence of the rank-one measurement state $\ket{w_r}$ corresponding to outcome $r$. For the effective Kraus operator $M_r = E^TW_r^\dagger$, the maximum success probability of faithful teleportation for outcome $r$ is 
\begin{equation}
\label{eq:single success probability}
   \frac{1}{4}\left[(1+\bar{E}_c\bar{E}_rx_r)-\sqrt{(1+\bar{E}_c\bar{E}_rx_r)^2-E_c^2E_r^2}\right],
\end{equation}
where $\bar{E}_c = \sqrt{1-E_c^2},~ \bar{E}_r = \sqrt{1-E_r^2}$. 
Here, $x_r=\mathbf{u}\cdot \mathbf{n}_r\in[-1,1]$ quantifies the relative alignment between the channel and measurement bases, where $\mathbf{u}$ and $\mathbf{n}_r$ are the Bloch directions of the reduced single-qubit operators $A:=E^*E^T$ and $B_r:= W_r^\dagger W_r$, respectively.
\end{Thm}

This result shows that, in the qubit case, the success probability is determined by three ingredients: channel entanglement $E_c$, measurement entanglement $E_r$, and relative alignment $x_r$ between their Bloch directions. In particular, $E_c$ and $E_r$ enter symmetrically through the factor $E_cE_r$, and the alignment dependence appears through the term $\bar{E}_c\bar{E}_rx_r$. Thus, whenever either the channel or the measurement becomes maximally entangled, the alignment dependence disappears.

The theorem also identifies relative alignment as an additional performance-determining factor in the partially entangled regime. For fixed values $E_c$ and $E_r$, different relative orientations between the channel and measurement Bloch directions can lead to different success probabilities. Therefore, in the qubit case, optimizing faithful teleportation requires not only increasing the amount of entanglement in both resources but also choosing their relative basis alignment appropriately.
The alignment parameter $x_r$ represents a local-basis degree of freedom and can therefore be adjusted without additional entangling gates or ancillary entangled resources. For fixed concurrence, the channel direction $\mathbf{u}$ and measurement direction $\mathbf{n_r}$ can be varied through local basis rotations during channel preparation and before the joint measurement, respectively, or equivalently through a redefinition of the corresponding local frames.

We now extend the analysis to higher-dimensional systems ($d>2$). In this case, entanglement alone does not fully determine the singular-value spectrum of the effective Kraus operators, and therefore a complete closed-form characterization analogous to the qubit case is generally not available. Nevertheless, one can still derive analytic bounds that capture how measurement entanglement constrains teleportation performance. To isolate the role of the joint measurement, we consider the maximally entangled channel in dimension $d$.

\begin{Thm}
\label{thm:high-dim}
Assume that the shared quantum channel is maximally entangled in dimension $d$. Then the maximum success probability satisfies
    \begin{equation}
        \frac{1}{d}\sum_r t_r \le P_{\text{succ}}^{\max} \le \frac{1}{d^2}\sum_r E_r
    \end{equation}
where $E_r$ denotes the entanglement of $\ket{w_r}$ quantified by the $G$-concurrence \cite{Gour2005}, and $t_r \in [0,1/d]$ is the unique solution of
\begin{equation}
t_r\left(\frac{1-t_r}{d-1}\right)^{d-1} = \left(\frac{E_r}{d}\right)^d.
\end{equation}
\end{Thm}

The upper bound is saturated when the joint measurement forms a maximally entangled basis.
The lower bound is tight for each measurement outcome when the singular values of $W_r$ take the form
\begin{equation} 
        \lambda_0 =\cdots = \lambda_{d-2} = \sqrt{\frac{1-t_r}{d-1}},\quad \lambda_{d-1}=\sqrt{t_r}\ .
    \end{equation}
Importantly, both bounds increase monotonically with the entanglement of the joint measurement, demonstrating that enhanced measurement entanglement systematically improves teleportation performance even for higher dimensions in this maximally entangled-channel setting.
Detailed derivations and saturation conditions are provided in Appendix~\ref{sec: high dim}.

\section{Applications} 
\label{sec:applications}

In this section, we apply our framework to several representative examples to investigate how the entanglement of joint measurements determines teleportation performance and how the protocol can be optimized. We first consider the elegant joint measurement, which serves as an illustrative theoretical example highlighting the roles of both the joint measurement and the quantum channel in quantum teleportation. We then turn to coherent-error models in realistic implementations of BSM, where our framework quantitatively captures how imperfect entangling operations deform the measurement basis and how their effects can be mitigated through optimized reversal strategies.

\subsection{Elegant joint measurement}

We consider the elegant joint measurement (EJM)~\cite{Tavakoli2021} as Alice's joint measurement in the teleportation protocol. It is defined by a set of four rank-one projectors $\{\ket{w_r(t)}\bra{w_r(t)}\}_{r=0}^3$, parameterized by $0\le t\le \pi/2$, given by 
\begin{equation}
    \begin{aligned}
    \ket{w_0(t)} & = \frac{1}{2}\left(e^{-\frac{\pi i}{4}}\ket{00}+p^{-}(t)\ket{01}+ p^+(t)\ket{10}+e^{-\frac{3\pi i}{4}}\ket{11}\right),\\
    \ket{w_1(t)} & = \frac{1}{2}\left(e^{\frac{3\pi i}{4}}\ket{00}+p^{-}(t)\ket{01}+ p^+(t)\ket{10}+e^{\frac{\pi i}{4}}\ket{11}\right),\\
    \ket{w_2(t)} & = \frac{1}{2}\left(e^{\frac{\pi i}{4}}\ket{00}-p^{+}(t)\ket{01}- p^-(t)\ket{10}+e^{\frac{3\pi i}{4}}\ket{11}\right),\\
    \ket{w_3(t)} & = \frac{1}{2}\left(e^{-\frac{3\pi i}{4}}\ket{00}-p^{+}(t)\ket{01}- p^-(t)\ket{10}+e^{-\frac{\pi i}{4}}\ket{11}\right),
    \end{aligned}
\end{equation}
where $p^{\pm}(t) = (1\pm e^{-it})/\sqrt{2}$. 
The EJM is an iso-entangled measurement so that all four basis states possess the same amount of entanglement, quantified by $E_M(t) = \sqrt{1-(3/4)\cos^2t}$ independent of the outcome $r$. The geometrical structure of the EJM is captured by the Bloch directions $\{\mathbf{n}_r\}$ of the reduced operators $W_r^\dagger W_r$. The corresponding Bloch vectors have common radius $\sqrt{1-E_M^2}=(\sqrt{3}/2)\cos t$ and form a regular tetrahedron in the Bloch sphere, as illustrated in Fig.~\ref{fig:EJM bloch sphere}. In particular, at $t=0$ the tetrahedron attains its maximal radius $\sqrt{3}/2$, while at $t=\pi/2$ it collapses to the origin and the EJM becomes locally equivalent to a BSM.

For fixed measurement entanglement $E_M$ and channel entanglement $E_c$, Theorem 1 reduces the optimization of the teleportation success probability to an optimization over the channel Bloch direction $\mathbf{u}$. For the EJM, the optimal alignment is achieved when the channel Bloch direction is antiparallel to one of the measurement Bloch directions, i.e.,~$\mathbf{u} = -\mathbf{n}_r$, as illustrated in Fig.~\ref{fig:EJM bloch sphere}(b). 
The optimization is performed to maximize the total success probability summed over all outcomes, rather than for a specific outcome. Due to the symmetry of the EJM, any choice of channel Bloch direction antiparallel to one of the measurement directions yields the same success probability.
This provides a concrete geometric realization of the fact that, in the partially entangled regime, teleportation performance depends not only on the amount of entanglement but also on the relative alignment between the channel and the measurement bases. To make this explicit, we consider a family of quantum channels parametrized by $0\le s\le \pi/2$, 
\begin{equation}
    \ket{\Phi(s)}_{ab}  = \frac{1}{2}\left(e^{-\frac{\pi i}{4}}\ket{00}+p^{-}(s)\ket{01}+ p^+(s)\ket{10}+e^{-\frac{3\pi i}{4}}\ket{11}\right).
\end{equation}
This family consists of channels whose local basis has been chosen so that the reduced operator has Bloch direction $\mathbf{u}=-\mathbf{n}_0$ and concurrence $E_c(s) = \sqrt{1-(3/4)\cos^2 s}$. Such a basis choice can be implemented by local rotations during channel preparation. For this choice, the tetrahedral symmetry gives $x_0=-1$ and $x_1=x_2=x_3=1/3$, leading to the following total optimized success probability:
\begin{equation}
\label{eq:EJM success probability}
    P_{\text{succ}}^{\max} = 1-\frac{1}{4}\left[\sqrt{(1-X)^2-(E_cE_M)^2}+\sqrt{(3+X)^2-9(E_cE_M)^2}\right],
\end{equation}
where $X:=\sqrt{(1-E_M^2)(1-E_c^2)}$. As shown in Fig.~\ref{fig:EJM noisy channel}, 
the success probability increases monotonically with both $E_c$ and $E_M$, approaching unity where both the channel and the measurement become maximally entangled.
If the experimental setup fixes the relative basis between the channel and the EJM, the total success probability should instead be obtained by summing Eq.~\eqref{eq:single success probability} over the experimentally realized values $x_r=\mathbf u\cdot\mathbf n_r$. 
In this sense, Eq.~\eqref{eq:EJM success probability} and Fig.~\ref{fig:EJM noisy channel} represent the optimized benchmark under local basis control.

\begin{figure}[t]
\centering
\includegraphics[width=8.5cm]{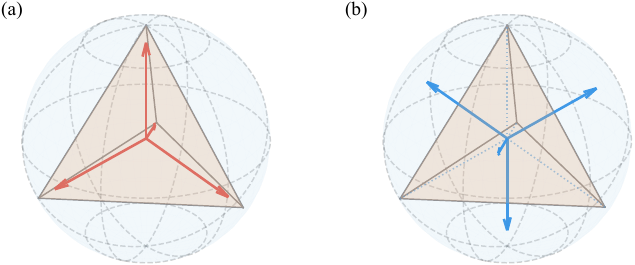}
\caption{Geometric structure of the elegant joint measurement (EJM) in the reduced Bloch-sphere representation. (a) The four Bloch directions $\{\mathbf n_r\}_{r=0}^3$ form a regular tetrahedron and the corresponding reduced Bloch vectors have common radius $\sqrt{1-E_M^2}$. (b) Optimal channel Bloch directions for fixed $E_c$ and $E_M$. The success probability of faithful teleportation is maximized when the channel Bloch direction is antiparallel to one of the measurement directions, $\mathbf u=-\mathbf n_r$, with Bloch-vector radius $\sqrt{1-E_c^2}$.}
\label{fig:EJM bloch sphere}
\end{figure}

\begin{figure}
\centering
\includegraphics[width=8.5cm]{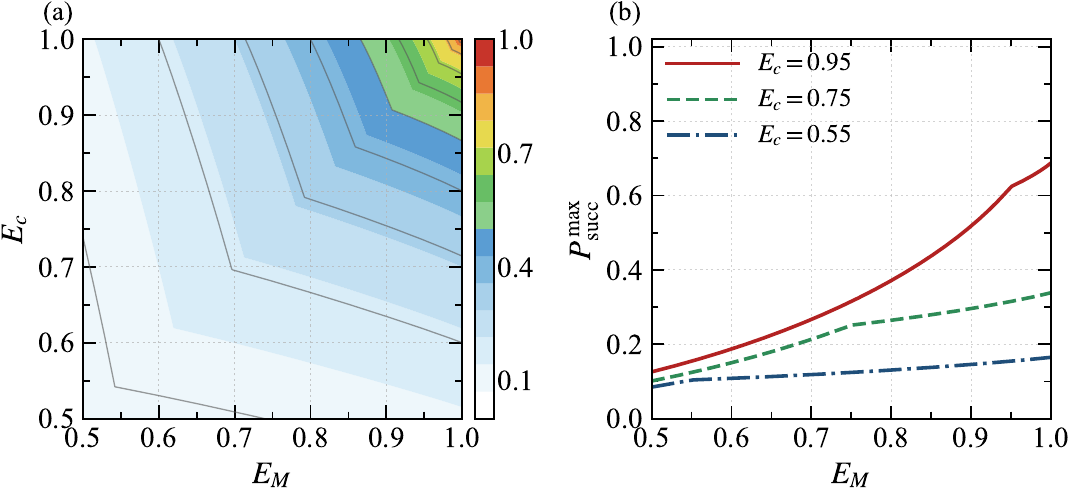}
\caption{Maximum success probability $P_\text{succ}^{\max}$ for faithful teleportation with the elegant joint measurement under the optimal channel-measurement alignment. (a) Contour plot of $P_\text{succ}^{\max}$ versus the measurement entanglement $E_M$ and the channel entanglement $E_c$. The symmetry of the contours under $E_M$ and $E_c$ indicates that, for the EJM, the two entanglement resources have an equal impact on the optimized success probability. (b) $P_\text{succ}^{\max}$ versus $E_M$ for several fixed values of $E_c$.}
\label{fig:EJM noisy channel}
\end{figure}

\begin{figure}[t]
\centering
\includegraphics[width=8.5cm]{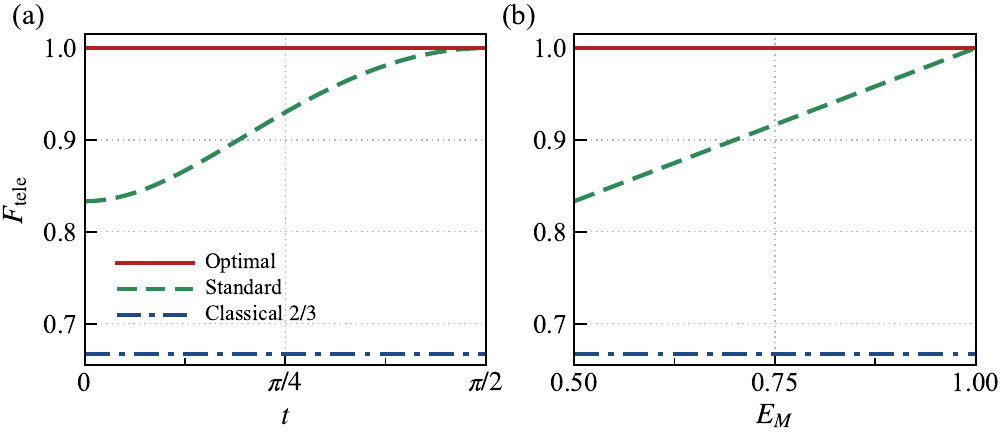}
\caption{Teleportation fidelity for the elegant joint measurement with a maximally entangled channel. 
(a) Fidelity $F_\text{tele}$ versus the EJM parameter $t$. In the standard protocol, the average fidelity increases from $5/6$ at $t=0$ to unity at the BSM limit $t=\pi/2$, whereas the optimized MR protocol achieves $F_\text{tele}=1$ for all $t$ by heralding. The horizontal line marks the classical limit $2/3$. (b) The same data plotted versus the measurement entanglement $E_M$.}
\label{fig:EJM standard fidelity}
\end{figure}

\begin{figure}[t]
\centering
\includegraphics[width=8.5cm]{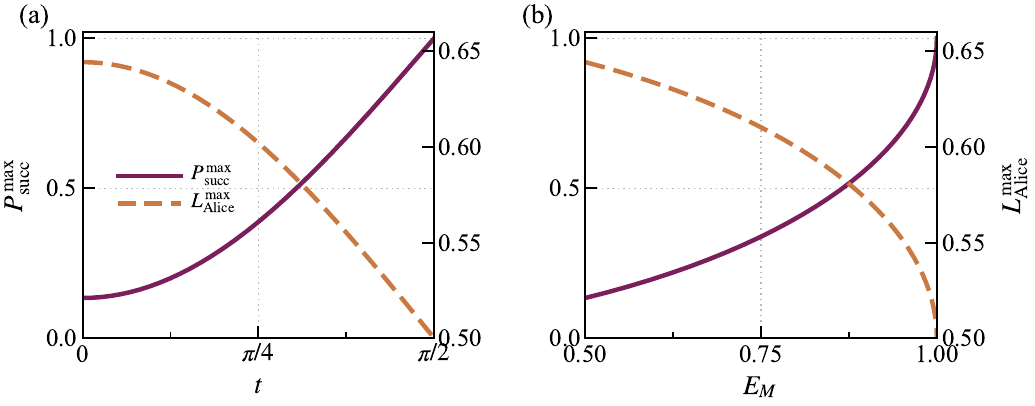}
\caption{Trade-off between the success probability and information leakage for the elegant joint measurement with a maximally entangled channel. (a) Maximum success probability $P_\text{succ}^{\max}$ and maximum information leakage to Alice $L_\text{Alice}^{\max}$ versus the EJM parameter $t$. (b) The same quantities plotted versus the measurement entanglement $E_M$. The two quantities satisfy $6L_\text{Alice}^{\max}+P_\text{succ}^{\max}=4$, demonstrating saturation of the no-cloning bound in the MR framework.}
\label{fig:EJM trade off}
\end{figure}


We next restrict our attention to a maximally entangled channel in order to isolate the effect of the joint measurement. For outcome $r=0$, the Kraus operator of the overall quantum instrument $\mathbf{M}_{\bar{a}\to b}$ is given by
\begin{equation}
    M_0(t)_ = \frac{1}{2\sqrt{2}}\begin{pmatrix} e^{\pi i/4} & (p^+(t))^* \\ (p^-(t))^* & e^{3\pi i/4}\end{pmatrix} \ .
\end{equation}
The remaining Kraus operators $M_r(t)$ are related by single-qubit Pauli conjugations,
\begin{equation}
    M_1 = -\sigma_zM_0\sigma_z,\quad  M_2  = -\sigma_xM_0\sigma_x,\quad M_3  = \sigma_yM_0\sigma_y.
\end{equation}
so that all $M_r(t)$ have the same singular values 
\begin{equation}
\label{eq:EJM_singular_values}
    \lambda_0(t) = \frac{\sqrt{2+\sqrt{3}\cos{t}}}{2\sqrt 2},\quad \lambda_1(t) = \frac{\sqrt{2-\sqrt{3}\cos{t}}}{2\sqrt 2}.
\end{equation}

We can then optimize the reversing operation performed by Bob within the MR framework, as given in Eq.~\eqref{eq:optimal reversing}. For outcome $r=0$, the optimal reversing operator has the form
\begin{equation}
    R_{0}(t)=\kappa(t)\begin{pmatrix}e^{-\pi i/4} & (p^+(t))^*\\ (p^-(t))^* & e^{-3\pi i/4}\end{pmatrix},\quad \kappa(t) = \frac{4\sqrt{2}\lambda_1(t)}{3-e^{2it}}
\end{equation}
The operators for the remaining outcomes are related by Pauli conjugations,
\begin{equation}
    R_1 = -\sigma_zR_0\sigma_z,\quad  R_2  = -\sigma_xR_0\sigma_x ,\quad R_3  = \sigma_yR_0\sigma_y.
\end{equation}
Notably, an arbitrary unknown input state can be faithfully recovered for all outcomes $r$ upon successful reversal, for all values of $t$, as
\begin{equation}
    R_r(t)M_r(t)\ket{\phi} = \lambda_{1}(t)\ket{\phi},\quad \forall \ket{\phi}.
\end{equation}
This is in contrast to the standard teleportation protocol, where the reversing operation is restricted to unitary corrections. In Fig.~\ref{fig:EJM standard fidelity}(a) and (b), we compare the maximum attainable teleportation fidelities of the optimized protocol and the standard protocol as functions of $t$ and the joint measurement entanglement $E_M$, respectively. We observe that, in the standard protocol, the maximum teleportation fidelity is approximately 0.83 at $t=0$ and increases monotonically with the entanglement of the joint measurement, reaching unity only when $E_M=1$. More precisely, we can calculate the average teleportation fidelity of the standard protocol as 
$F_\text{tele}(t)=2/3+\sqrt{4-3\cos^2 t}/6$.
In contrast, the optimized protocol within our framework achieves unit heralded teleportation fidelity, $F_{\text{tele}} = 1$ for all values of $t$, irrespective of the measurement entanglement.

This improvement is achieved at the expense of the success probability. The maximum success probability of the optimized protocol is given by $P_\text{succ}^{\max}(t) = 1-(\sqrt{3}/2)\cos t$. As analyzed in Sec.~\ref{sec:framework}, the corresponding information leakage to Alice is $L_\text{Alice}^{\max} (t) = 1/2+(\sqrt{3}/12)\cos t$. Figure~\ref{fig:EJM trade off}(a) and (b) show the dependence of the success probability and the information leakage on $t$ and the joint measurement entanglement $E_M$, respectively. This clearly demonstrates a trade-off between the success probability and the information leakage, reflecting the fundamental no-cloning bound, while maintaining unit teleportation fidelity.

\begin{figure}[t]
\centering
\includegraphics[width=8.5cm]{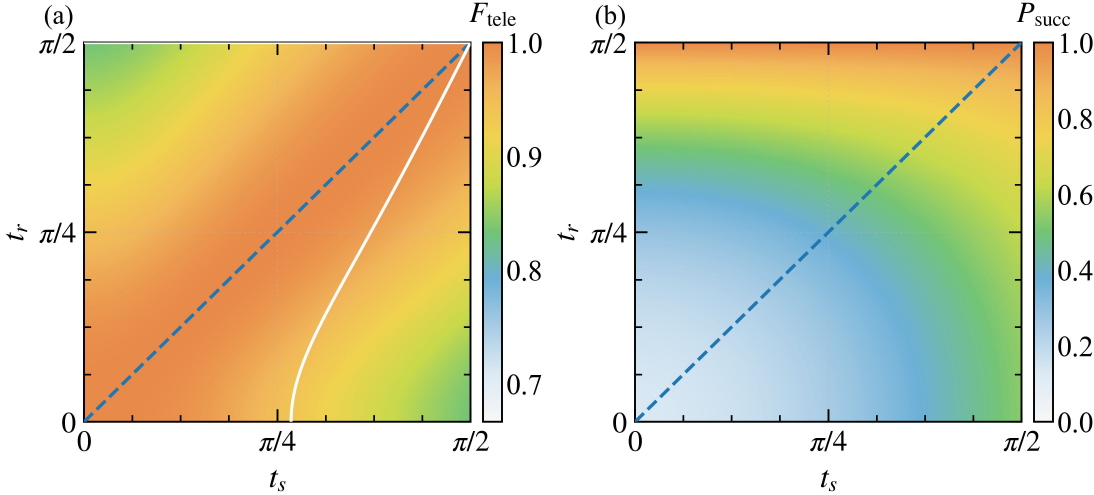}
\caption{Effect of imperfect parameter knowledge on the EJM-based protocol with a maximally entangled channel.
The actual measurement parameter is denoted by $t_s$, while $t_r$ is the value assumed by Bob in constructing the reversing operation.
(a) Teleportation fidelity $F_\mathrm{tele}(t_s,t_r)$.
(b) Corresponding success probability $P_\mathrm{succ}(t_s,t_r)$. The blue dashed diagonal represents perfect parameter knowledge, $t_r=t_s$, along which successful reversal achieves unit fidelity.
In panel (a), the region above the white solid curve indicates where the proposed protocol outperforms standard teleportation using adaptive unitary correction alone.}
\label{fig:EJM_mismatch}
\end{figure}


We further examine the effect of imperfect parameter knowledge in the EJM-based protocol. Let $t_s$ be the actual EJM parameter realized in Alice's joint measurement and $t_r$  the value used by Bob to construct the reversing operation. A mismatch may arise from imperfect measurement characterization or parameter drift between characterization and teleportation. We assume a maximally entangled channel to isolate this effect. Using the singular values in Eq.~\eqref{eq:EJM_singular_values} and defining $r(t)=\lambda_1(t)/\lambda_0(t)$, the teleportation fidelity and the corresponding success probability are given by
\begin{equation}
    F_{\mathrm{tele}}(t_s,t_r) = \frac{1}{3} \left[1+\frac{\left(\lambda_0(t_s)r(t_r)+\lambda_1(t_s)\right)^2}{\left(\lambda_0(t_s)r(t_r)\right)^2+\lambda_1(t_s)^2}\right],
\end{equation}
and
\begin{equation}
    P_{\mathrm{succ}}(t_s,t_r) = 2\left[\left(\lambda_0(t_s)r(t_r)\right)^2+\lambda_1(t_s)^2 \right].
\end{equation}
For perfect parameter knowledge, $t_s=t_r$, the protocol achieves unit fidelity with the optimal success probability derived above. As shown in Fig.~\ref{fig:EJM_mismatch}, the fidelity remains robust against moderate parameter mismatch.
The blue dashed diagonal represents the perfect-knowledge condition $t_r=t_s$, while the white solid curve in the fidelity panel bounds the region in which the MR protocol outperforms standard teleportation for the same actual measurement parameter $t_s$. This favorable region contains the diagonal and extends over a broad range of parameter mismatches, demonstrating that the MR protocol can retain its fidelity advantage even with imperfect parameter knowledge.
The success-probability panel shows 
the corresponding heralding probability and hence the probabilistic cost of fidelity recovery.
When the actual EJM parameter is already close to the Bell-state-measurement limit resulting in only a small measurement-induced distortion, a mismatched reversing operation may overcorrect or undercorrect it. This accounts for the regions in which standard teleportation outperforms the MR protocol.

As a result, using the EJM, we can explicitly observe how teleportation performance depends on the entanglement of the joint measurement. While the standard protocol remains deterministic, it suffers from fidelity loss when the measurement entanglement is reduced, whereas the optimized protocol within our framework preserves unit teleportation fidelity at the cost of a reduced (heralded) success probability.

\subsection{Coherent-error models in realistic BSM implementations}
\label{sec:applications_realistic}

As a realistic application of our framework, we first consider the implementation of BSM in superconducting qubits, where two-qubit entangling gates are commonly realized using the cross-resonance (CR) gate \cite{Paraoanu2006, Rigetti2010, Chow2011}. In the CR platform, the effective two-qubit Hamiltonian generally contains several Pauli terms, but after standard echo calibration the desired entangling interaction is dominantly of $ZX$ type, yielding a CNOT-equivalent operation up to local rotations \cite{Magesan2020, Sundaresan2020}.
In practice, however, residual coherent errors remain during gate execution, most often static or drive-induced $ZZ$ couplings. We therefore model an imperfect BSM as a $ZX$-based CNOT implementation perturbed by a coherent $ZZ$ error during the entangling evolution \cite{Sundaresan2020, Tripathi2019}.

The ideal entangling operation can be written as $ZX_{\pi/2} = e^{-i(\pi/4)Z_{\bar{a}}X_{a}}$, where $\bar{a}$ and $a$ denote the control and target qubits, respectively. In the presence of a coherent $ZZ$ error, the entangling operation becomes
\begin{equation}
    U_\text{ent}(t) = \exp \left[-i\left(\frac{\pi}{4}Z_{\bar{a}}X_{a}+tZ_{\bar{a}}Z_{a}\right)\right],
\end{equation}
where $0\le t<\sqrt{3}\pi/4$ quantifies the coherent error strength and $t=0$ corresponds to the ideal case which is equivalent to CNOT gate up to local unitary operations. We model the imperfect BSM using the imperfect CNOT gate $U_\text{CNOT}(t) = U_\text{loc}U_\text{ent}(t)$.
The resulting imperfect BSM can be modeled by a deformed Bell basis $\{\ket{w_r(t)}\}_{r=0}^3$
\begin{equation}
    \begin{aligned}
    \ket{w_0(t)}  &= \frac{1}{2}\Bigr(\alpha(t)\ket{00}+\beta(t)\ket{01}-\beta(t)\ket{10}+\alpha(t)\ket{11}\Bigl),\\ 
    \ket{w_1(t)}& = \frac{1}{2}\Bigr(-\beta(t)^*\ket{00}+\alpha(t)^*\ket{01}+\alpha(t)^*\ket{10}+\beta(t)^*\ket{11}\Bigl),\\
    \ket{w_2(t)} & = \frac{1}{2}\Bigr(\alpha(t)\ket{00}+\beta(t)\ket{01}+\beta(t)\ket{10}-\alpha(t)\ket{11}\Bigr),\\
    \ket{w_3(t)}& = \frac{1}{2}\Bigr(-\beta(t)^*\ket{00}+\alpha(t)^*\ket{01}-\alpha(t)^*\ket{10}-\beta(t)^*\ket{11}\Bigr), 
    \end{aligned}
\end{equation}
where we denote $R(t):=\sqrt{\pi^2+16t^2}/4$, $\alpha(t)= \frac{\pi}{4R}\sin R+\cos R+i \frac{t}{R}\sin R$, and $\beta(t) = i\left(\frac{\pi}{4R}\sin R-\cos R\right)-\frac{t}{R}\sin R$. Within this interval, the coherent $ZZ$ error continuously reduces the entanglement of the implemented joint measurement. The concurrence of each measurement state is given by $E_M(t) = |\sin2R(t)|$, which decreases monotonically from $1$ at $t=0$ to $0$ as $t\to \sqrt{3}\pi/4$.

We first evaluate how the reduction of entanglement in the joint measurement, induced by coherent errors, affects the teleportation fidelity within the standard protocol. To isolate the effect of measurement imperfections, we consider a maximally entangled quantum channel. The average teleportation fidelity is then obtained as $F_\text{tele} (t) = (2+E_M(t))/3$, which attains unity only at $t=0$ and decreases monotonically as the coherent error increases. This demonstrates that, in realistic superconducting-gate implementations, coherent errors in the entangling operation directly degrade the performance of the standard deterministic teleportation protocol through the reduction of measurement entanglement.

We next apply our framework to optimize the teleportation protocol under the effect of the same imperfect BSM.
Since the underlying $ZX$-based implementation naturally selects the $\sigma_y$ basis, we consider the partially entangled quantum channel $\ket{\Phi(\varphi)} = \cos \varphi \ket{0_y0_y}+\sin\varphi\ket{1_y1_y}$,
for $0\le \varphi\le \pi/4$, where $\ket{0_y} = (\ket{0}+i\ket{1})/\sqrt{2}$ and $\ket{1_y} = (\ket{0}-i\ket{1})/\sqrt{2}$. The amount of entanglement of this state is given by $E_c(\varphi) = \sin 2\varphi$. 

The Kraus operators of the overall quantum instrument are given by
\begin{equation}
    \begin{aligned}
        M_0 & = \gamma(t)\cos\varphi\sin R(t)\ket{0_y}\bra{1_y}+\sin\varphi\cos R(t)\ket{1_y}\bra{0_y},\\
        M_1 & = \cos\varphi\cos R(t)\ket{0_y}\bra{0_y}-\gamma(t)^*\sin\varphi\sin R(t)\ket{1_y}\bra{1_y},\\
        M_2 & = \gamma(t)\cos\varphi\sin R(t)\ket{0_y}\bra{0_y}+\sin\varphi\cos R(t)\ket{1_y}\bra{1_y},\\
        M_3 & = \cos\varphi\cos R(t)\ket{0_y}\bra{1_y}-\gamma(t)^*\sin\varphi\sin R(t)\ket{1_y}\bra{0_y},\\
    \end{aligned}
\end{equation}
where $\gamma(t):= (\pi-4it)/\sqrt{\pi^2+16t^2}$.
The corresponding optimized reversing operators can then be obtained as
\begin{equation}
    \begin{aligned}
        R_0 & = \ket{0_y}\bra{1_y}+\tan\varphi\cot R(t)\,\gamma(t)^*\ket{1_y}\bra{0_y},\\
        R_1 & = \begin{cases}\tan\varphi\tan R(t)\ket{0_y}\bra{0_y}-\gamma(t)\ket{1_y}\bra{1_y},\quad 0\le t\le t_\varphi,\\ \ket{0_y}\bra{0_y}- \cot \varphi \cot R(t)\,\gamma(t)\ket{1_y}\bra{1_y},\quad t_\varphi\le t< \frac{\sqrt{3}\pi}{4},\end{cases}\\
        R_2 & = \tan\varphi\cot R(t)\,\gamma(t)^*\ket{0_y}\bra{0_y}+\ket{1_y}\bra{1_y},\\
        R_3 & = \begin{cases}-\gamma(t)\ket{0_y}\bra{1_y}+\tan\varphi\tan R(t)\ket{1_y}\bra{0_y},\quad 0\le t\le t_\varphi,\\ -\cot\varphi\cot R(t)\,\gamma(t)\ket{0_y}\bra{1_y}+\ket{1_y}\bra{0_y},\quad t_\varphi\le t<\frac{\sqrt{3}\pi}{4}, \end{cases}\\
    \end{aligned}
\end{equation}
where $t_\varphi = \sqrt{(\pi/2-\varphi)^2-(\pi/4)^2}$.

\begin{figure}[t]
\centering
\includegraphics[width=8.5cm]{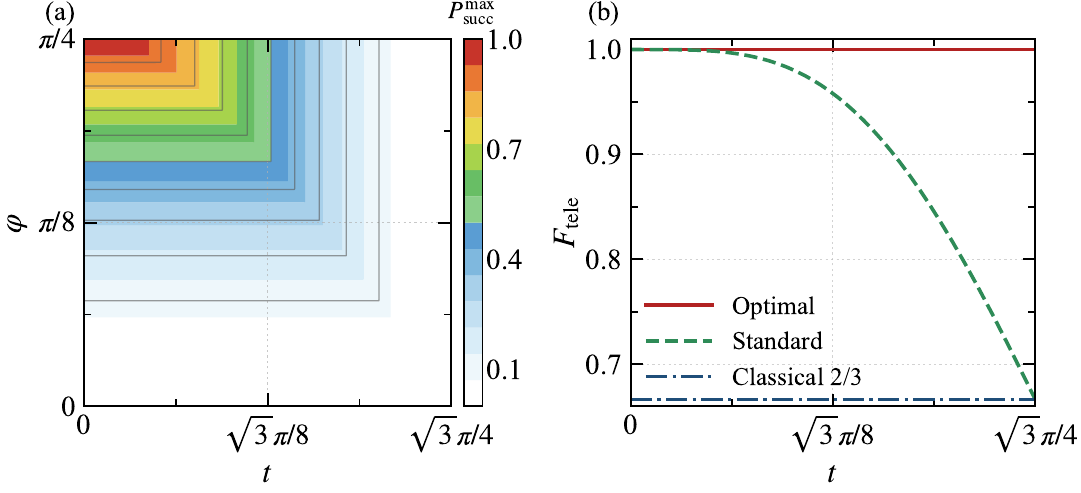}
\caption{Teleportation performance under a $ZX$-based BSM with $ZZ$-type coherent error. (a) Maximum success probability $P_\text{succ}^{\max}$ for faithful teleportation as a function of the error parameter $t$ and the channel parameter $\varphi$. (b) Teleportation fidelity $F_\text{tele}$ for a maximally entangled channel, comparing the optimized MR protocol with the standard protocol. While the standard fidelity decreases with $t$, the optimized protocol maintains unit fidelity for all $t$ upon successful heralded reversal at the probabilistic cost quantified by the success probability in panel (a).}
\label{fig:ZXerror}
\end{figure}


Notably, by applying the optimal reversing operation, unit teleportation fidelity can always be achieved with a nonzero success probability.
Applying Theorem 1, the maximum success probability for faithful teleportation is given by
\begin{equation}
P_\text{succ}^{\max} = 1-\max\{\cos2\varphi,|\cos2R(t)|\} 
\end{equation}
As shown in Fig.~\ref{fig:ZXerror}(a), the success probability is maximized at $t=0$ and $\varphi=\pi/4$ where both the measurement and the channel are maximally entangled, recovering deterministic and faithful teleportation. For a fixed $\varphi$, the success probability is highest near the ideal point $t=0$ and decreases as the error strength $t$ increases, since the coherent $ZZ$ error reduces the entanglement of the implemented joint measurement.
We further compare the maximum achievable teleportation fidelities of the optimized and standard protocols as a function of $t$ for a maximally entangled channel in Fig.~\ref{fig:ZXerror}(b).
As a result, unlike standard teleportation, the optimized protocol in our framework preserves unit teleportation fidelity for known coherent measurement error, while the imperfection appears instead as a reduction of the success probability. 

\begin{figure}[t]
\centering
\includegraphics[width=8.5cm]{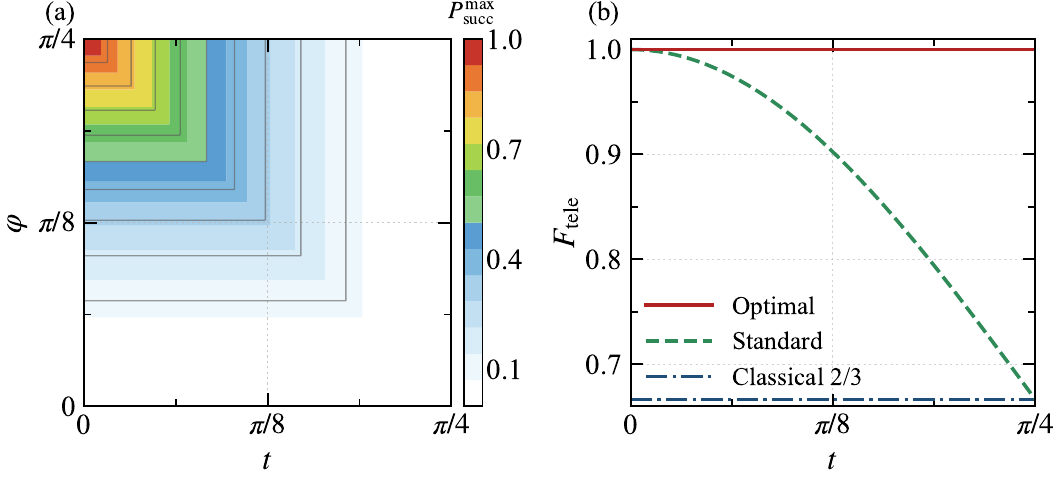}
\caption{Teleportation performance under $XX$-type coherent error in the BSM. (a) Maximum success probability $P_\text{succ}^{\max}$ for faithful teleportation as a function of the error parameter $t$ and the channel parameter $\varphi$. (b) Teleportation fidelity $F_\text{tele}$ for a maximally entangled channel, comparing the optimized MR protocol with the standard protocol. While the standard fidelity decreases with $t$, the optimized protocol maintains unit fidelity for all $t$ upon successful heralded reversal at the probabilistic cost quantified by the success probability in panel (a).}
\label{fig:XXerror}
\end{figure}

As another realistic application, we consider the implementation of a Bell-state measurement (BSM) based on an $XX$ Ising interaction, as introduced in Sec.~\ref{sec:effect of coherent error}. The entangling operation is given by
\begin{equation}
R_{XX}(t) = e^{-i(\pi/4-t)X_1X_2},
\end{equation}
where the ideal entangling operation is achieved at $t=0$, while deviations in the entangling time $0 < t \le \pi/4$ induce coherent errors. The resulting BSM corresponds to a rotated, time-dependent projection of the form given in Eq.~\eqref{eq:deformedBB}. 
Consequently, the entanglement of the joint measurement basis decreases monotonically with $t$, as shown in Fig.~\ref{fig:coherent teleportation}(b). In the standard deterministic teleportation protocol, this directly leads to a reduction in teleportation fidelity, as illustrated in Fig.~\ref{fig:coherent teleportation}(a).

By applying our framework, the teleportation protocol can be optimized to achieve faithful teleportation with a maximized success probability. To illustrate this explicitly in a general setting, we consider a partially entangled quantum channel of the form $\ket{\Phi(\varphi)}=\cos\varphi\ket{00}+\sin\varphi\ket{11}$ with $0\le\varphi\le\pi/4$, whose entanglement is given by $E_c(\varphi) =\sin2\varphi$.
The corresponding Kraus operators of the overall quantum instrument, $M_r = E^{T}W_r^\dagger$, are given by
\begin{equation}
    \begin{aligned}
        M_0 & =\cos\varphi\sin\theta\ket{0}\bra{0}-i\sin\varphi\cos\theta\ket{1}\bra{1},\\
        M_1 & =-i\cos\varphi\cos\theta\ket{0}\bra{1}+\sin\varphi\sin\theta\ket{1}\bra{0},\\
        M_2 & =\cos\varphi\cos\theta\ket{0}\bra{0}+i\sin\varphi\sin\theta\ket{1}\bra{1},\\
        M_3 & =i\cos\varphi\sin\theta\ket{0}\bra{1}+\sin\varphi\cos\theta\ket{1}\bra{0},
    \end{aligned}
\end{equation}
where $\theta=\pi/4-t$. From Eq.~\eqref{eq:optimal reversing}, the optimal reversing operators for each outcome are given by
\begin{equation}
    \begin{aligned}
        R_0 & = \begin{cases}\ket{0}\bra{0}+i\cot\varphi\tan\theta\ket{1}\bra{1},\quad \theta\le \varphi,\\ \tan\varphi\cot\theta\ket{0}\bra{0}+i\ket{1}\bra{1},\quad \theta> \varphi,\end{cases}\\
        R_1 & = \sigma_x(i\tan\varphi\tan\theta\ket{0}\bra{0}+\ket{1}\bra{1}),\\
        R_2 & = \tan\varphi\tan\theta\ket{0}\bra{0}-i\ket{1}\bra{1},\\
        R_3 & = \begin{cases}\sigma_x(-i\ket{0}\bra{0}+\cot\varphi\tan\theta\ket{1}\bra{1}),\quad \theta\le \varphi,\\ \sigma_x(-i\tan\varphi\cot\theta\ket{0}\bra{0}+\ket{1}\bra{1}),\quad \theta>\varphi.\end{cases}
    \end{aligned}
\end{equation}

Therefore, by performing the optimal reversing operation at Bob's side,
faithful teleportation is possible with unit teleportation fidelity upon the successful reversal.
Applying Theorem \ref{thm:success probability}, the maximum success probability is evaluated as
\begin{equation}
\label{eq:XX_success_prob}
    P_\text{succ}^{\max} = \begin{cases}2\sin^2(\frac{\pi}{4}-t),& t\ge \frac{\pi}{4}-\varphi\\ 2\sin^2\varphi,& t\le \frac{\pi}{4}-\varphi\end{cases},
\end{equation}
or equivalently, $P_\text{succ}^{\max} = 1-\sqrt{1-\min\{E_c,E_M\}^2}$ with $E_M(t) = \cos 2t$. Figure~\ref{fig:XXerror}(a) shows the maximum success probability as a function of the error parameter $t$ and the channel parameter $\varphi$.
This result shows that the teleportation performance is determined by the smaller of the two resources, namely, the entanglement of the channel and that of the joint measurement.
In contrast, in the standard teleportation protocol, where Bob’s correction is restricted to unitary operations, the average teleportation fidelity is given by $F_\text{tele}(\varphi,t) = (2+\sin2\varphi\cos2t)/3$, which is strictly less than unity away from the ideal point. Figure~\ref{fig:XXerror}(b) compares the maximum achievable teleportation fidelities of the standard and optimized protocols for a given coherent error in the BSM and a maximally entangled channel.

We also examine the effect of imperfect parameter knowledge in a realistic BSM implementation. To avoid repeating analogous analyses for both coherent-error models, we focus on the XX-Ising implementation as a representative example. To distinguish the entangling angle from the coherent-error parameter $t$, we define $\theta=\frac{\pi}{4}-t$. Let $\theta_s$ denote the angle realized in Alice's BSM and $\theta_r$ the value used by Bob to construct the reversing operation. Their mismatch may arise from imperfect BSM characterization or drift in the interaction strength. For a maximally entangled channel, the teleportation fidelity and success probability are
\begin{equation}
    F_{\mathrm{tele}}(\theta_s,\theta_r)
    = \frac{1}{3}\left[
    1+
    \frac{\left(\cos\theta_s\tan\theta_r+\sin\theta_s\right)^2}
    {\cos^2\theta_s\tan^2\theta_r+\sin^2\theta_s}
    \right],
\end{equation}
and
\begin{equation}
    P_{\mathrm{succ}}(\theta_s,\theta_r)
    =\tan^2\theta_r\cos^2\theta_s+\sin^2\theta_s .
\end{equation}
When $\theta_r=\theta_s$, these expressions yield $F_{\mathrm{tele}}=1$ and the optimal success probability in Eq.~\eqref{eq:XX_success_prob}. Figure~\ref{fig:XX_mismatch} shows that the fidelity recovery remains robust against small parameter mismatches. The blue dashed diagonal represents $\theta_r=\theta_s$, while the white solid curve in the fidelity panel bounds the region in which the heralded fidelity of the MR protocol exceeds that of standard teleportation for the same actual angle $\theta_s$. This favorable region extends over a finite range around the diagonal, with the corresponding heralding cost shown in the success-probability panel. When $\theta_s$ is close to $\pi/4$, however, Alice's BSM is already nearly ideal and the residual distortion is small. A mismatched reversing operation may then overcorrect or undercorrect this distortion, accounting for the regions in which standard teleportation outperforms the MR protocol.

\begin{figure}[t]
\centering
\includegraphics[width=8.5cm]{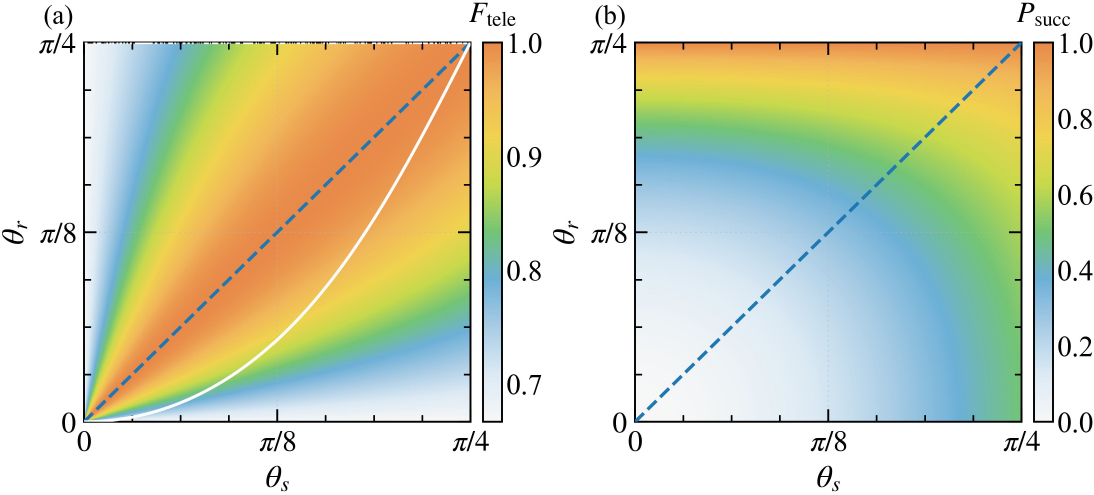}
\caption{Effect of imperfect parameter knowledge in the XX-Ising implementation of the Bell-state measurement with a maximally entangled channel. The entangling angle is defined as $\theta=\pi/4-t$, where $t$ is the coherent-error parameter. The actual measurement angle is denoted by $\theta_s$, while $\theta_r$ is the value used by Bob to construct the reversing operation. (a) Teleportation fidelity $F_{\mathrm{tele}}(\theta_s,\theta_r)$. (b) Corresponding success probability $P_{\mathrm{succ}}(\theta_s,\theta_r)$. The blue dashed diagonal represents perfect parameter knowledge, $\theta_r=\theta_s$, along which the reversal achieves unit fidelity conditioned on success. In panel (a), the region above the white solid curve indicates where the proposed protocol achieves a higher heralded fidelity than standard teleportation using adaptive unitary correction alone.}
\label{fig:XX_mismatch}
\end{figure}


\section{Conclusion} 
\label{sec:conclusion}

We have investigated how the entanglement of partially entangled joint measurements induced by coherent errors determines the performance of quantum teleportation. We showed that, within the standard deterministic teleportation protocol, any reduction in measurement entanglement inevitably leads to a degradation of teleportation fidelity, even when the shared quantum channel remains maximally entangled. To overcome this limitation, we employed a generalized teleportation framework based on measurement reversal (MR), which enables probabilistic recovery of the input state. Within this framework, we demonstrated that unit-fidelity teleportation can always be restored for pure-state inputs with a nonzero success probability, regardless of the degree of measurement entanglement, provided that the imperfection is known. Moreover, we derived explicit relations that quantitatively connect the entanglement of the channel, the entanglement of the joint measurement, and the success probability of faithful teleportation, and illustrated these results through both theoretical and realistic models.

We have shown that the entanglement of the joint measurement plays a role that is fully comparable to that of the channel entanglement in determining teleportation performance. In particular, we showed that the maximum success probability of faithful teleportation is jointly constrained by the two resources, revealing a clear operational interpretation of measurement entanglement. For qubit systems, we derived a complete analytical characterization that decomposes the performance into entanglement and alignment contributions, while for higher-dimensional systems with a maximally entangled channel, we derived bounds that monotonically improve with measurement entanglement. Although a closed-form expression is generally unavailable for $d>2$, our analysis demonstrates that the same framework extends naturally to higher dimensions and captures the essential structure governing teleportation performance beyond qubits.

While we have mainly focused on coherent errors that deform the measurement basis, the proposed framework can be further extended to incorporate stochastic noise and decoherence affecting either the joint measurement or the quantum channel. In such scenarios, the achievable teleportation fidelity becomes lower than unity even after optimization, as additional noise sources fundamentally limit the target fidelity~\cite{Lee2021}. 
The unit fidelity guarantee also relies on accurate knowledge of the measurement-induced Kraus operators. As illustrated by the EJM and XX-Ising examples, imperfect parameter knowledge removes this exact guarantee, although the MR protocol can retain a fidelity advantage over a finite range of parameter mismatch.
Nevertheless, the MR framework still provides a systematic way to mitigate coherent components of the error and to optimize the trade-off between fidelity and success probability. This suggests that our approach can serve as a general tool for analyzing and improving teleportation protocols under realistic noise environments.

Quantum teleportation is a key primitive in a wide range of quantum technologies, including gate teleportation for universal quantum computation~\cite{Gottesman1999,Wan19}, modular quantum computing architectures~\cite{Llewellyn20, Chou18}, and distributed quantum networks~\cite{Briegel1998,Raussendorf2001,Kimble2008,Lee20}. In these settings, coherent errors arising in entangling operations are ubiquitous and can significantly degrade performance if the standard deterministic protocol is used without modification. We stress that our results show that such limitations can be overcome by adopting a heralded strategy, where unit-fidelity teleportation is maintained while imperfections are shifted into a controllable success probability that can be optimized. The proposed method can be implemented without substantial modifications to existing hardware, as it relies only on classical post-processing and local filtering operations. From a practical perspective, our protocol is therefore complementary to other error-mitigation methods and may be combined with them, as it specifically targets coherent measurement distortions through heralded local filtering. We therefore expect that our framework will provide a practical route toward improving teleportation-based quantum protocols in current experimental platforms.

\section*{ACKNOWLEDGMENTS}
This research was funded by National Research Foundation of Korea (RS-2022-NR068812, RS-2024-00442762), Institute of Information \& Communications Technology Planning \& Evaluation (RS-2025-02263264), Korea Institute of Science and Technology (2E33541), and Global Partnership Program of Leading Universities in Quantum Science and Technology (RS-2025-02317602).

\appendix

\section{Post-measurement state after joint measurement}
\label{A:post-measurement}
We derive Eq. \eqref{eq:overall karus operator}, which identifies the effective Kraus operator induced by the shared quantum channel and the joint measurement.
For notational simplicity, we consider a pure input state $\rho=\ket{\phi}\bra{\phi}$. Since the induced operation is linear in the input state, the result extends immediately to an arbitrary mixed input state $\rho_{\bar{a}}$, yielding $\tilde{\rho}_b^r = M_r\rho_{\bar{a}}M_r^\dagger$.
Write the shared quantum channel, the rank-one measurement state, and the input state as 
\begin{equation}
    \begin{aligned}
       \ket{\Phi}_{ab} &= \sum_{i,j=0}^{d-1}E_{ij}\ket{i}_a\otimes \ket{j}_b\\ \ket{w_r}_{\bar{a}a}&=\sum_{m,n=0}^{d-1}(W_r)_{mn}\ket{m}_{\bar{a}}\ket{n}_{a}\\ \ket{\phi}_{\bar{a}} &= \sum_{p=0}^{d-1}\phi_p\ket{p}_{\bar{a}} \ . 
    \end{aligned}
\end{equation}
Then the post-measurement state at Bob's mode is 
\begin{equation}
    (\bra{w_r}_{\bar{a}a}\otimes I_b)(\ket{\phi}_{\bar{a}}\otimes \ket{\Phi}_{ab})  = \sum_{p,i,j}\phi_p E_{ij}(\bra{w_r}_{\bar{a}a}\ket{p}_{\bar{a}}\ket{i}_a)\ket{j}_b \ .
\end{equation}
Using
\begin{equation}
    \bra{w_r}_{\bar{a}a}\ket{p}_{\bar{a}}\ket{i}_a=\sum_{m,n}(W_r)^*_{mn}\braket{m|p}\braket{n|i} = (W_r)^*_{pi},
\end{equation}
we obtain
\begin{equation}
    (\bra{w_r}_{\bar{a}a}\otimes I_b)(\ket{\phi}_{\bar{a}}\otimes \ket{\Phi}_{ab})=\sum_{p,i,j}\phi_p E_{ij}(W_r)_{pi}^*\ket{j}_b\ .
\end{equation}
Rearranging the sums gives the following.
\begin{equation}
\begin{aligned}
(\langle w_r|_{\bar a a}\otimes I_b)(|\phi\rangle_{\bar a}\otimes|\Phi\rangle_{ab})
&=
\sum_j
\left(
\sum_{i,p} E_{ij}(W_r)^\ast_{pi}\phi_p
\right)
|j\rangle_b\\
&=
\sum_j (E^{T}W_r^\dagger|\phi\rangle)_j |j\rangle_b
\\&=
E^{T}W_r^\dagger|\phi\rangle_b \ .
\end{aligned}
\end{equation}

\section{Proof of Theorem \ref{thm:success probability}}
\label{sec: proof of thm1}

For qubits, Eq. \eqref{eq:optimal reversing} implies that the success probability for outcome $r$ is given by $p_{\text{succ},\phi}^r=\sigma_{\min}^2(M_r)$ where $M_r =E^{T}W_r^\dagger$. Thus, the problem reduces to computing the smallest singular value of $M_r$. Let $H_r=M_rM_r^\dagger$, which is a positive $2\times 2$ matrix. Then
\begin{equation}
\label{eq:mini eigenvalue}
    \sigma_{\min}^2(M_r) = \frac{\mathrm{Tr}H_r -\sqrt{(\mathrm{Tr}H_r)^2-4\det H_r}}{2}\ .
\end{equation}
We first evaluate the determinant
\begin{equation}
    \det H_r = |\det M_r|^2 = |\det E|^2|\det W_r|^2 =\frac{E_c^2E_r^2}{16},
\end{equation}
where we used that, for a normalized two-qubit pure state with coefficient matrix $X$, the concurrence is $2|\det X|$.

For the trace, define $A:=E^* E^{T}$ and $B_r:=W_r^\dagger W_r$. Since both $A$ and $B_r$ are single-qubit density operators, they admit Bloch representations:
\begin{equation}
    A= \frac{1}{2}\left(I+\bar{E}_c\mathbf{u}\cdot \boldsymbol{\sigma}\right),\quad B_r= \frac{1}{2}\left(I+\bar{E}_r\mathbf{n}_r\cdot \boldsymbol{\sigma}\right),
\end{equation}
with Bloch radii $\bar{E}_c = \sqrt{1-E_c^2}$ and $\bar{E}_r = \sqrt{1-E_r^2}$. 
Here, $\mathbf{u}$ and $\mathbf{n}_r$ are Bloch directions and $\boldsymbol{\sigma} = (\sigma_x,\sigma_y,\sigma_z)$ is the Pauli vector.
Hence,
\begin{equation}
    \mathrm{Tr}H_r = \mathrm{Tr} (AB_r ) = \frac{1}{2}\left(1+\bar{E}_c\bar{E}_rx_r\right),
\end{equation}
where $x_r:=\mathbf{u}\cdot\mathbf{n}_r\in[-1,1]$.
Substituting the above expressions for $\det H_r$ and $\mathrm{Tr}H_r$ into Eq. \eqref{eq:mini eigenvalue} yields Eq. \eqref{eq:single success probability}, which completes the proof.

\section{Role of entanglement of joint measurement in higher dimensions}
\label{sec: high dim}

For $d>2$, unlike the qubit case, the entanglement of a bipartite pure state does not uniquely determine the singular-value spectrum of its coefficient matrix. As a result, the maximum success probability cannot be expressed solely in terms of a single entanglement parameter. Nevertheless, for a maximally entangled channel, we can derive tight bounds in terms of the entanglement of the joint measurement.

Assume that the shared channel is maximally entangled, such that $E = U/\sqrt{d}$ for some unitary $U$. 
Let $\{\lambda_i^r\}_{i=0}^{d-1}$ be the singular values of $W_r$, ordered as $\lambda_0^r\ge \cdots \ge \lambda_{d-1}^{r}$. 
Since multiplication by a unitary does not affect singular values, the singular values of $M_r = E^{T} W_r^\dagger$ are given by $\{\lambda_i^r / \sqrt{d}\}$.
From Eq.~\eqref{eq:optimal reversing}, the maximum success probability is therefore $P_{\text{succ}}^{\max} = \sum_r(\lambda_{d-1}^r)^2/d$.
To quantify the entanglement of $\ket{w_r}$, we use the $G$-concurrence \cite{Gour2005}
\begin{equation}
    E_r = d\left(\prod_{i=0}^{d-1}\lambda_i^r\right)^{2/d} \ .
\end{equation}

\begin{proof}[Proof of Theorem \ref{thm:high-dim}]
The upper bound follows from the elementary inequality that the minimum of a set of nonnegative numbers cannot exceed their geometric mean:
\begin{equation}
\label{eq:qudit upper bound}
    (\lambda_{\min}^r)^2 \le \left(\prod_{i=0}^{d-1} (\lambda_i^r)^2\right)^{1/d}=\frac{E_r}{d},
\end{equation}
Summing over $r$ gives
\begin{equation}
    P_{\text{succ}}^{\max}=\frac{1}{d}\sum_r\left(\lambda_{d-1}^r\right)^2\le\frac{1}{d^2}\sum_rE_r \ .
\end{equation}
In particular, if the joint measurement forms a maximally entangled basis, then $\lambda_0^r = \cdots = \lambda_{d-1}^r=1/\sqrt{d}$ and $E_r = 1$ for all $r$. Hence Eq. \eqref{eq:qudit upper bound} is saturated for all $r$, and thus the upper bound is saturated.

For the lower bound, let $x_i^r = (\lambda_i^r)^2$ and $s_r:=x_{d-1}^r$.
Since $\sum_ix_i^r =1$, we have 
\begin{equation}
    x_0^r+\cdots +x_{d-2}^r=1-s_r\ .
\end{equation}
Applying the AM-GM inequality to $x_0^r,\ldots,x_{d-2}^r$, we obtain 
\begin{equation}
\label{eq:qudit AM-GM}
\prod_{i=0}^{d-2}x_i^r\le \left(\frac{1-s_r}{d-1}\right)^{d-1},
\end{equation}
and thus
\begin{equation}
\prod_{i=0}^{d-1}x_i^r=s_r\prod_{i=0}^{d-2}x_i^r\le s_r\left(\frac{1-s_r}{d-1}\right)^{d-1}\ .
\end{equation}
Using $\prod_ix_i^r = (E_r/d)^d$, this becomes
\begin{equation}
\label{eq:qudit t_r bound}
    \left(\frac{E_r}{d}\right)^d\le s_r\left(\frac{1-s_r}{d-1}\right)^{d-1}\ .
\end{equation}

Now, we define the function $g(t):= t\left(\frac{1-t}{d-1}\right)^{d-1}$ on the interval $0\le t\le 1/d$.
Its derivative is $g'(t) = (1-t)^{d-2}(1-dt)/(d-1)^{d-1}$, which is nonnegative on $0\le t\le 1/d$ and strictly positive on $0\le t<1/d$.
Hence $g(t)$ is monotonic on this interval.
Since $g(0)=0$ and $g(1/d) = (1/d)^d$, there exists a unique $t_r\in[0,1/d]$ satisfying 
\begin{equation}
    g(t) = \left(\frac{E_r}{d}\right)^d\ .
\end{equation}
Hence, Eq.~\eqref{eq:qudit t_r bound} implies $s_r\ge t_r$.
Therefore,
\begin{equation}
\label{eq:qudit lower bound}
    P_\text{succ}^{\max}=\frac{1}{d}\sum_r s_r \ge \frac{1}{d} \sum_r t_r \ .
\end{equation}

Moreover, the lower bound in Eq. \eqref{eq:qudit lower bound} is attained when equality holds in Eq. \eqref{eq:qudit AM-GM} for every outcome $r$. By the equality condition of the AM-GM inequality, this occurs exactly when $x_0^r=\cdots =x_{d-2}^r = (1-s_r)/(d-1)$. In this case, 
\begin{equation}
    \left(\frac{E_r}{d}\right)^d = \prod_{i=0}^{d-1}x_i^r = s_r\left(\frac{1-s_r}{d-1}\right)^{d-1} =g(s_r),
\end{equation}
which implies $s_r= t_r$. Therefore, the lower bound is tight for each measurement outcome when the singular values of $W_r$ are $\sqrt{(1-t_r)/(d-1)},\ldots, \sqrt{(1-t_r)/(d-1)}, \sqrt{t_r}$.
\end{proof}

Since $g$ is strictly increasing on $[0,1/d]$, its inverse $g^{-1}$ is also strictly increasing. Therefore, $t_r= g^{-1}((E_r/d)^d)$ increases monotonically with $E_r$. The upper bound also increases monotonically in $E_r$.
Thus, for a maximally entangled channel, both bounds improve monotonically as the entanglement of the joint measurement increases.
For $d=2$, the lower bound becomes exact with Theorem~\ref{thm:success probability} for a maximally entangled channel.
For $d=3$, we get $t_r = \frac{2}{3}+\frac{2}{3}\cos\left(\frac{1}{3}\arccos(2E_r^3-1)+\frac{2\pi}{3}\right)$, while for higher dimensions $d\ge 4$, $t_r$ can be obtained numerically.


\end{document}